\documentclass[english,a4paper,12pt,pointlessnumbers,normalheadings,bibtotoc,liststotoc]{article}%

\RequirePackage[latin1]{inputenc}
\RequirePackage{verbatim}
\RequirePackage{footmisc}	
\RequirePackage{enumerate}				
\RequirePackage{babel}
\RequirePackage[numbers]{natbib}
\RequirePackage{graphicx,overpic,picins,subfigure}
\RequirePackage{xcolor}
\RequirePackage{framed}
\RequirePackage{amsmath,amssymb,amsthm,latexsym,exscale}
\RequirePackage{dlfltxbcodetips}
\RequirePackage{bbm}
\RequirePackage{xspace}
\RequirePackage{version}
\RequirePackage{setspace}
\RequirePackage{ifpdf}
\ifpdf
	\RequirePackage[T1]{fontenc}
	\RequirePackage{ae}
\fi
\RequirePackage{hyperref}
\hypersetup{%
	  plainpages  	= 	false,
	  colorlinks	=	true,                         
	  linkcolor		=	black,                       
	  citecolor		=	black,                      
	  urlcolor		=	black,                  
	  linkcolor		=	black,                   
	  citecolor		=	black,                    
	  urlcolor		=	black,                    
	  bookmarksopen	=	true,
	  pdfpagemode	=	None,                  
}

\renewcommand{\familydefault}{cmss}
\numberwithin{equation}{section}					
\newcommand{\Beq}{\begin{equation}}
\newcommand{\Eeq}{\end{equation}}
\newcommand{\BeqO}{\begin{equation*}}
\newcommand{\EeqO}{\end{equation*}}
\newcommand{\NN}{\mathbbm{N}}

\newcommand{\RR}{\mathbbm{R}}

\newcommand{\OO}{\mathcal{O}}											
\newcommand{\azk}{n}
\newcommand{\drk}{d}
\newcommand{\auf}[2]{\ensuremath{ #1,\ldots , #2 }}							
\renewcommand {\le}{\left}													
\newcommand {\ri}{\right}													
\newcommand{\qmbox}[1]{\quad\text{#1}\quad}									
\newcommand{\dd}{\mathrm{d}}												
		
\newcommand{\tab}[2]{\ensuremath{ \frac{\dd #1}{\dd #2} }}					

\newcommand{\abs}[1]{\ensuremath{ \left\vert #1 \right\vert }}				
\newcommand{\ee}{\varepsilon}												
\newcommand {\eh}{{\textstyle \frac{1}{2}}}

\newcommand{\vuo}[2]{  \ensuremath{ v_{#1}^{(#2)}} } 
\newcommand{\dvuo}[2]{ \ensuremath{ \Delta  v_{#1}^{(#2)}} } 
\newcommand{\Auo}[3]{  \ensuremath{ #1_{#2}^{(#3)}} } 

\newtheoremstyle{saetze}{10pt}{10pt}{\itshape}{}{\bfseries}{:}{.5em}{}
\theoremstyle{saetze}
\newtheorem{defi}{Definition}[section]
\newtheorem{theorem}[defi]{Theorem}
\newtheorem{lemma}[defi]{Lemma}

\addto{\captionsngerman}{%
  }

\newcommand{\beq}{\begin{equation}}
\newcommand{\eeq}{\end{equation}}
\newcommand{\Leq}[1]{\phantomsection \label{#1}\end{equation}}
\newcommand{\beqn}{\begin{eqnarray}}
\newcommand{\eeqn}{\end{eqnarray}}
\newcommand{\beqno}{\begin{eqnarray*}}
\newcommand{\eeqno}{\end{eqnarray*}}

\renewcommand {\l}{\left}

\newcommand {\bN}{{\mathbb N}}

\newcommand{\bem}{\l(\! \begin{array}}
\newcommand{\eem}{\end{array}\!\ri)}
\newcommand{\bsm}{\left(\begin{smallmatrix}} 
\newcommand{\esm}{\end{smallmatrix}\right)}  
\renewcommand {\max}{{\textup{\textrm{max}}}}

\begin{document}
\title {Elastic Scattering of Point Particles With Nearly Equal Masses}
\author{%
Andreas Knauf %
\thanks{%
Department Mathematik, Universit\"{a}t Erlangen-N\"{u}rnberg, %
Bismarckstr.\ $1 \eh$, %
91054 Erlangen, Germany, %
e-mail: knauf@mi.uni-erlangen.de %
}%
\and
Markus Stepan
}
\date{September 2011}
\maketitle
\begin{abstract}
We show that for $n$ billiard particles on the line 
there exist three open sets in the product of phase space 
and the space  of their masses, such that these particles 
exhibit exactly $n-1$, $\bsm n\\ 2\esm$ respectively 
$\bsm n+1\\ 3\esm$ collisions. These open sets intersect any
neighborhood of the diagonal in mass space.
\end{abstract}
\section{Introduction}
We consider a system of the $\azk$ hard spheres, that move on a 
straight line. The elastic collision between two spheres is to be 
the only interaction, being defined by conservation of 
energy and momentum.

Estimates of the number of collisions in hard ball systems, and more 
generally, in semidispersing
billiards, have been studied for a long time.
But it was only in 1998 when Burago, Ferleger and Kononenko \cite{burago98}  
could show the following statement,  using tools from metric geometry: 
 
\begin{quote}
In a system of $ \azk$ hard spheres, that move in $ \RR^{\drk} $, 
the total number of collisions is bounded above uniformly in the 
initial conditions.
\end{quote}

All known upper bounds on the collision number of a general system 
of hard spheres increase superexponentially in the number of particles 
$ \azk $, even for equal masses.

One of the few systems that can be solved exactly is that for 
equal masses and $d=1$. 
For $d=1$ without loss of generality one can assume the 
particles to be pointlike.
For equal masses scattering then just exchanges the labels
of the particles, thus generically leading to exactly 
${\azk \choose 2} $ collisions.

In this paper we are going to study of hard sphere systems 
in $d=1$, having approximately equal masses. 
Chen could show in the 
papers \cite{chen07,chen09} that subject to certain 
requirements for the masses, the upper quadratic bound of 
$ {\azk \choose 2} $ remains true. Now we show that 
additionally to these cases, there are
open sets in the product of phase space 
and the space $(0,\infty)^{n}$ of masses, such that these particles 
exhibit exactly $n-1$, $\bsm n\\ 2\esm$ respectively 
$\bsm n+1\\ 3\esm$ collisions. These three open sets intersect any
neighborhood of the diagonal in mass space.

This shows that --- concerning the number of collisions --- the 
case of exactly equal masses is nongeneric.

\section{Nearly Equal Masses}

The position of the $k$--th (pointlike) particle at time $t$ is
denoted by $q_k(t)$, and we assume initial conditions 
with $q_k(0)<q_{k+1}(0)\quad(k=1,\ldots, n-1)$.\\
We begin with a trivial observation, valid only for $d=1$.
\begin{lemma}
If for $\drk=1$, all $\azk\in\NN$  and any mass distribution $(\auf{m_1}{m_\azk})$
the number of collisions is strictly smaller than $\azk-1$, then all velocities
are equal, so that no collision occurs.
\end{lemma}
\begin{proof}
Under the above assumption there is a particle, with number $k<\azk$, 
not being involved in
any collision with particle number $k+1$, so that
$t\mapsto q_k(t)$ is convex and $t\mapsto q_{k+1}(t)$
is concave. But as $q_k(t)<q_{k+1}(t)$, both functions must be affine,
with velocities $v_{k+1}=v_k$.
So particle number $k$ experiences no
collision with particle number $k-1$ and particle number $k+1$
experiences no collision with particle number $k+2$.
An iteration of the argument shows the assertion.
\end{proof}

For $\drk=1$ and $\azk$ equal masses the open set of initial
conditions 
\Beq
	U=U(\azk):=\{(q,v)\in \RR^n\times \RR^\azk \mid q_k< q_l\mbox{ and }
	v_k> v_l \mbox{ for } 1\leq k< l\leq n\}
	\label{eq:ma}
\Eeq
leads to exactly ${\azk \choose 2}$ two-body collisions (if one continues 
multi-body collisions in a way that the {\em set} of velocities $v_k$ does not change), 
and all these collisions occur at positive times.

By Chen's result \cite{chen07} for $\drk=1$ and any $\ee>0$ there exists an
open set of $\azk$ masses $(\auf{m_1}{m_\azk})\in (1-\ee,1+\ee)^\azk$,
leading to at most ${\azk \choose 2}$ collisions.

This is complemented by the following result: 

\begin{theorem}\label{thm:maintheorem}
For $\drk =1$, all $\azk \in \NN$ and for all $\ee > 0$ there exist non-empty open subsets 
$V_i:=W_i \times U_i \subset (1-\ee,1+\ee)^\azk \times (\RR^\azk \times \RR^\azk)$ $(i=1,2,3)$  
of the extended phase space such that for any initial data in the set the number of collisions
\begin{enumerate}[1)]\setlength{\itemindent}{10mm}
	\item  on $V_1$ equals $\azk-1 ={\azk-1 \choose 1}$,
	\item  on $V_2$ equals ${\azk \choose 2}$, and
	\item  on $V_3$ equals ${\azk+1 \choose 3}$.
\end{enumerate}
\end{theorem}
\begin{proof}
\textit{Preliminary note:} 
As the velocities are unchanged between collisions, we
denote the velocity of the $i$--th particle between the 
$(\ell-1)$--th and $\ell$--th collision by $\vuo{i}{\ell-1}$.
If the $\ell$--th collisions involves particles $i$ and $i+1$, then
\Beq
\vuo{i}{\ell}  =  
	{\textstyle \frac{ (m_{i}-m_{i+1})\vuo{i}{\ell-1} + 2m_{i+1}\vuo{i+1}{\ell-1}  } {m_{i}+m_{i+1} }}
	\qmbox{,}
	\vuo{i+1}{\ell}  =  {\textstyle \frac{ (m_{i+1}-m_{i})\vuo{i+1}{\ell-1} + 2m_{i}\vuo{i}{\ell-1}  }
	{m_{i}+m_{i+1} }}.
\label{vv}
\Eeq
We define for $i=\auf{1}{\azk-1}$ the velocity differences by 
$\dvuo{i}{\ell}:=\vuo{i+1}{\ell}-\vuo{i}{\ell}$. 
During a collision between particles $i$ and $i+1$ we have
\Beq
\begin{gathered}
	\dvuo{i}{\ell}=-\dvuo{i}{\ell-1} 	
	\qmbox{,}	
	\dvuo{k}{\ell}= \dvuo{k}{\ell-1} \text{ if }\abs{i-k}>1 ,  \label{eq:gesch_diff}\\
	\dvuo{i-1}{\ell} = \dvuo{i-1}{\ell-1}+ {\textstyle \frac{2m_{i+1}}{m_{i}+m_{i+1}}} \dvuo{i}{\ell-1} \qmbox{and}
	\dvuo{i+1}{\ell} = \dvuo{i+1}{\ell-1}+ {\textstyle \frac{2m_{i}}{m_{i}+m_{i+1}}} \dvuo{i}{\ell-1} . 
\end{gathered}
\Eeq
Later we will apply the following functions, indexed by $k\in\bN$, to
quotients of adjacent masses:
\Beq
\begin{aligned}
 	g_1(x)&:=g(x):= 2x-1  , \   
 	& g_{k+1} &:= g \circ g_{k} 
 	\text{, thus }  
 	& g_{k} (x) &= 2^k(x-1)+1 .\\ 
 	f_1(x)&:=f(x):= \frac{1+x}{3-x} , \   
 	&f_{k+1} &:= f \circ f_{k} 
 	\text{, thus }   
 	& f_{k} (x) &= \frac{k(x-1)-2x}{k(x-1)-2} .
\end{aligned}\label{eq:schranken_fkten}
\Eeq
For all these functions
$\lim_{x\searrow 1} f_{k} (x) = \lim_{x\searrow 1} g_{k} (x) = 1$
 and  
\BeqO
	  f_{k} (x) = g_{k} (x) = {\textstyle \frac{k+2^{k+1}-2}{k}} 
	  \qquad \mbox{ for }x={\textstyle \frac{k+2-2^{1-k}}{k}}.
\EeqO
The function $f_{k}$ has a pole at $(k+2)/k$, and 
because
\BeqO
	\tab{}{x}f_{k} (x) = \frac{4}{(k(x-1)-2)^2} > 0 \qquad \mbox{ for all } x,k
\EeqO
$f_{k}$ is strictly increasing in the interval $\big(1,(k+2)/k\big)$. 
In addition $f_{k}$ is convex in that interval, because
\BeqO
	\tab{^2}{x^2}f_{k} (x) = -\frac{8k}{(k(x-1)-2)^3} > 0 
	\qmbox{if and only if} x < \frac{k+2}{k}.
\EeqO
By the validity of the inequalities $\frac{k+1}{k} \leq \frac{k+2-2^{1-k}}{k} \leq \frac{k+2}{k}$
for 
 $k\in \NN$, we see that for all $x$ in the {\em sub}interval $\big(1,(k+1)/k\big)$
the following inequalities are also true:
\Beq
	g_{k} (x) > f_{k} (x) = x + \OO\le((x-1)^2\ri) > x.
	\label{ineq:gfx}
\Eeq

\noindent
\textit{Case 1)} 
It is the idea of the proof to find conditions on the masses and the 
initial velocities, so that the sequence of collisions is 
$(1,2)$, $(2,3)$, $(3,4)$,\ldots, $(\azk-1,\azk)$.

We consider for $\delta:= (1+\ee)^{1/(n-1)}$ the non-empty open set
$W_1=W_1(n,\delta):=$
\BeqO
	\le\{(\auf{m_1}{m_n})\in (1-\ee,1+\ee)^n\mid\forall 
	k=\auf{2}{n}: {\textstyle \frac{m_{k}}{m_{k-1}}}\in(1,\delta)\ri\}.
\EeqO
We prove the assertion by induction on the number $\azk$ of particles. 

The case $\azk=2$ is trivial.
Now we assume that we found a neighbourhood $U_1(\azk-1)$ for 
the first $\azk-1$ particles. 
Then for all initial conditions 
$x \in U_1(\azk-1)$ there exists a time $T(x) \in (0,\infty)$, so
that no collisions occur after $T(x)$. 
$T(x)$ is chosen to be continuous on $U_1(\azk-1)$. 
Further there is a continuous function 
$Q:U_1(\azk-1) \rightarrow \RR$, such that 
$Q(x) > \max\{q_{\azk-1}(t)\mid t\in[0,T(x)]\}$.
 
We now consider the open set $\tilde{U}_1(\azk):=$
\BeqO
	\le\{ (x,q_\azk, v_\azk)\in U_1(\azk-1)\times \RR^2 
	\mid 
	v_\azk < v_{\azk-1},\ q_\azk>Q(x),\ q_\azk+ v_\azk T(x)>Q(x) \ri\}.
\EeqO

According to our induction assumption the last collision took place 
between the particle $\azk-2$ and $\azk-1$, 
and the next one should occur 
between the particle $\azk-1$ and $\azk$. 
We can see from the equation \eqref{eq:gesch_diff} of the 
preliminary note 

\begin{align*}
\dvuo{\azk-2}{\azk-1} 
&=  
\dvuo{\azk-2}{\azk-2} + \frac{2m_{\azk}}{m_{\azk-1} + m_{\azk}}\dvuo{\azk-1}{\azk-2}\\
& =  
-\dvuo{\azk-2}{\azk-3} + \frac{2m_{\azk}}{m_{\azk-1} + m_{\azk}} \le( \dvuo{\azk-1}{\azk-3} 
+ \frac{2m_{\azk-2}}{m_{\azk-2}+m_{\azk-1}}\dvuo{\azk-2}{\azk-3} \ri)\\
& = 
{\textstyle \frac{4
m_{\azk-2}m_{\azk}-(m_{\azk-2}+m_{\azk-1})(m_{\azk-1}+m_{\azk})}{(m_{\azk-2}+m_{\azk-1})(m_{\azk-1}+m_{\azk})}}
\dvuo{\azk-2}{\azk-3}+
{\textstyle {\frac{2m_{\azk}}{m_{\azk-1}+m_{\azk}}}} \dvuo{\azk-1}{\azk-3}
	.
\end{align*}

A necessary condition is that particle $\azk-2$ and $\azk-1$ have no 
collision once again so that $\dvuo{\azk-2}{\azk-1}\geq 0$. 
Thus, inserting the definition of the $\dvuo{r}{s}$, 
and noting $\vuo{\azk-1}{\azk-3} = \vuo{\azk-1}{0}$ and 
$\vuo{\azk}{\azk-3} = \vuo{\azk}{0}$, we see that
\[
0 \leq
{\textstyle \frac{4
m_{\azk-2}m_{\azk}-(m_{\azk-2}+m_{\azk-1})(m_{\azk-1}+m_{\azk})}{(m_{\azk-2}+m_{\azk-1})(m_{\azk-1}+m_{\azk})}}
\le(\vuo{\azk-1}{0} - \vuo{\azk-2}{\azk-3} \ri)
+ {\textstyle {\frac{2m_{\azk}}{m_{\azk-1}+m_{\azk}}}} 
\le( \vuo{\azk}{0} - \vuo{\azk-1}{0} \ri)\!\!  .
\]
With the above choice of $\tilde{U}_1(\azk)$ 
and by solving the last inequality, we have
\begin{eqnarray}
	0 &\!\!>\!\!& 
	\dvuo{\azk-1}{0} \geq \le( 
	{\textstyle
	\frac{3m_{\azk-1}-m_{\azk-2}}{2(m_{\azk-2}+m_{\azk-1})}} +
	{\textstyle \frac{m_{\azk-1}}{2m_{\azk}}} -1 \ri)
	\vuo{\azk-1}{0}
	+\le( {\textstyle\frac{3m_{\azk-2}-m_{\azk-1}}{2(m_{\azk-2}+m_{\azk-1})}} -
	{\textstyle\frac{m_{\azk-1}}{2m_{\azk}}} \ri)\vuo{\azk-2}{\azk-3} 
	\nonumber\\
	&\!\!=\!\!& 
	\le(
	\frac{m_{\azk-1}-3m_{\azk-2}}{2(m_{\azk-2}+m_{\azk-1})} + 
	\frac{m_{\azk-1}}{2m_{\azk}} \ri) 
	\le(  \vuo{\azk-1}{0} - \vuo{\azk-2}{\azk-3} \ri). \label{eq:ungl_n-1}
\end{eqnarray}
As $\vuo{\azk-2}{\azk-3} > \vuo{\azk-1}{\azk-3} = \vuo{\azk-1}{0}$ 
(otherwise there would be no $(\azk-2)$--th collision), the inequality is satisfied if and only if
\Beq
	 m_{\azk} \leq 
	 \frac{m_{\azk-1}(m_{\azk-1}+m_{\azk-2})}{3m_{\azk-2}-m_{\azk-1}} 
	 \qmbox{for} 
	 3m_{\azk-2}>m_{\azk-1} . 
	 \label{eq:massen_rekursion_1}
\Eeq
{With $f$ from} \eqref{eq:schranken_fkten} the quotient 
$\frac{m_{\azk}}{m_{\azk-1}}$ has a recursive inequality
\BeqO
	\frac{m_{\azk}}{m_{\azk-1}} \leq 
	\frac{m_{\azk-1}+m_{\azk-2}}{3m_{\azk-2}-m_{\azk-1}} = 
	\frac{1 + \frac{m_{\azk-1}}{m_{\azk-2}}}{3-\frac{m_{\azk-1}}{m_{\azk-2}}} = 
	f\le( \frac{m_{\azk-1}}{m_{\azk-2}} \ri) .
\EeqO
By {inequality (\ref{ineq:gfx}) of} the preliminary note we can find two initial conditions 
$m_1$ and $m_2$ in the interval $\big(1,(\azk+1)/\azk\big)$, such that
\BeqO
	1 < 
	\frac{m_{\azk}}{m_{\azk-1}} < 
	f_{\azk-2}\le( \frac{m_{2}}{m_{1}}\ri)= 
	\frac{(\azk-2)\le(\frac{m_2}{m_1}-1\ri)-2\frac{m_2}{m_1}}{(n-2)\le(\frac{m_2}{m_1}-1\ri)-2} < 
	\delta =(1+\ee)^{1/(\azk-1)} .
\EeqO
We choose a velocity $\vuo{\azk}{0}$ in the interval
\BeqO
	\left( 
	\left(
{\textstyle\frac{3m_{\azk-1}-m_{\azk-2}}{2(m_{\azk-2}+m_{\azk-1})}} + 
	{\textstyle\frac{m_{\azk-1}}{2m_{\azk}}} \right)\vuo{\azk-1}{0} 
	+\left(
{\textstyle\frac{3m_{\azk-2}-m_{\azk-1}}{2(m_{\azk-2}+m_{\azk-1})}} - 
	{\textstyle\frac{m_{\azk-1}}{2m_{\azk}}}
	\right)\vuo{\azk-2}{\azk-3}\ ,\ \vuo{\azk-1}{0} 
	\right) .
\EeqO
This is possible, since the length of the interval is positive, 
according to the formula \eqref{eq:ungl_n-1} and to the 
condition \eqref{eq:massen_rekursion_1} on the masses. 

By the induction hypothesis we found open neighbourhoods 
\[\auf{\Auo{V}{1}{0} = 	\big(
	\Auo{\underline{v}}{1}{0},
	\Auo{\overline{v}}{1}{0}\big)\  }{\ \Auo{V}{\azk-1}{0} =
	\big(\Auo{\underline{v}}{\azk-1}{0},
	\Auo{\overline{v}}{\azk-1}{0}\big) }
\]
and thus we found a neighbourhood 
$
	\Auo{V}{\azk-2}{\azk-3} = \big( 
	\Auo{\underline{v}}{\azk-2}{\azk-3}, 
	\Auo{\overline{v}}{\azk-2}{\azk-3} \big)
$ 
for $\vuo{\azk-2}{\azk-3}$, too.
We can find a neighbourhood for $\vuo{\azk}{0}$ if and only if the length of the interval for $\vuo{\azk}{0}$ is greater than zero. 
Therefore we must have
\BeqO
	\min \left\lbrace \vuo{\azk-2}{\azk-3}-\vuo{n-1}{0}\ |\ 
	\vuo{\azk-2}{\azk-3}\in \Auo{V}{\azk-2}{\azk-3},\vuo{\azk-1}{0}
	\in \Auo{V}{\azk-1}{0} \right\rbrace = 
	\Auo{\underline{v}}{\azk-2}{\azk-3} - \Auo{\overline{v}}{\azk-1}{0} > 0.
\EeqO
One can achieve this after possibly reducing the size of 
the interval $ \Auo{V}{\azk-1}{0} $.

Now we find an explicit configuration space neighbourhood of 
the $\azk$--th particle. We have to choose the lower limit of the 
interval by
\BeqO
	\Auo{\underline{q}}{\azk}{0} \geq 
	\frac{ \Auo{\overline{q}}{\azk-1}{0} - 
	\Auo{\underline{q}}{\azk-2}{0} }{ \Auo{\underline{v}}{\azk-2}{0} 
	- \Auo{\overline{v}}{\azk-1}{0} }\left( \Auo{\underline{v}}{\azk-1}{0} -
	 \Auo{\overline{v}}{\azk}{0} \right) + \Auo{\underline{q}}{\azk-1}{0} .
\EeqO
So we have the required sequence of collisions. 
We can choose the upper limit $\Auo{\overline{q}}{\azk}{0}$ greater than $\Auo{\underline{q}}{\azk}{0}$.

Therefore we  found a non-empty open set 
$V_1:=W_1 \times U_1 \subset \RR_+^n \times \RR^\azk \times \RR^\azk$
in extended phase space, leading to 
$\azk-1={\azk-1 \choose 1}$ collisions.
\bigskip 

\noindent
\textit{Case 2)} 
Here we use a transversality argument, perturbing the case of equal masses $m_1=\ldots=m_n=1$.
For equal masses there is an open bounded set 
$ \tilde{U}_2 \subset U(\azk)\subset \RR^{\azk} \times \RR^{\azk} $ 
(with $U(\azk)$ from (\ref{eq:ma}))
of initial conditions in phase space with the following properties:
Both the minimal time between the (binary) collisions and the minimal collision angle 
\BeqO
\beta\; :=\; \min_{1\leq i<j\leq n}\arccos 
{\textstyle\frac{ 1+\vuo{i}{0}\vuo{j}{0} }
{ \sqrt{ 1+\big( \vuo{i}{0}\big)^2 }\sqrt{ 1+\big( \vuo{j}{0}\big)^2 }
}}\ >\ 0
\EeqO
in the extended 
configuration space $ \RR^\azk \times \RR_t $ are bounded from below by positive constants
for all $(q^{(0)},v^{(0)})\in \tilde{U}_2$.

The final positions and velocities for binary collisions depend continuously on the 
initial data and the masses, see (\ref{vv}).
Thus be uniform continuity on compacts, we find non-empty open neighborhoods $U_2\subset\tilde{U}_2$
and $W_2= (1-\ee,1+\ee)^n$ for which the same statement holds true for the
initial data in the subset $W_2\times U_2$ of extended phase space. In particular, the number of 
collisions  equals ${ n \choose 2}$.

\bigskip 

\noindent
\textit{Case 3)} 
We consider for $\delta:= (1+\ee)^{1/(\azk-1)}$ and 
$\delta'\in (1,\delta)$ the non-empty open set
$W_3=W_3(n,\delta,\delta'):=$
\BeqO
\le\{
	(\auf{m_1}{m_\azk})\in (1-\ee,1+\ee)^\azk\ \Big|\ 
	\forall k=\auf{2}{\azk}: {\textstyle\frac{m_{k}}{m_{k-1}}}\in(\delta',\delta)\ri\}.
\EeqO
We prove the assertion again by induction on the number of the particles $\azk$. 
The case for $\azk=2$ is simple, since then ${ n+1 \choose 3}=1$.

In analogy to Case 1) we assume that we have already found a neighbourhood 
$U_3(\azk-1)$ for the first $\azk-1$ particles. 
Moreover for any initial  condition $x \in U_3(\azk-1)$ 
there is a time $T(x) \in (0,\infty)$, such that there are no more 
collisions after that time $T(x)$. Further exists a continuous 
function $Q:U_3(\azk-1) \rightarrow \RR$, such that $Q(x) > \max\{q_{\azk-1}(t)\mid t\in[0,T(x)]\}$.

Now we consider an open set $\tilde{U}_3(\azk):=$
\BeqO
	\le\{ 
	(x,q_\azk,v_\azk)\in U_3(\azk-1)\times \RR^2 \mid 
	v_\azk < v_{\azk-1},\ q_\azk>Q(x),\ q_\azk+ v_\azk T(x)>Q(x) \ri\}.
\EeqO 
By these initial conditions particle $\azk-1$ and $n$ will hit after time $T$.
We will show, that for an appropriate subset 
$U_3(\azk) \subset \tilde{U}_3(\azk)$ 
there are exactly ${\azk \choose 2}$ collisions after time $T$. 
This then proves the inductive step, since 
\[{\textstyle { n \choose 3}+{ n \choose 2}={ n+1 \choose 3}}.\]
Assumed, that there is a collision after time $T(x)$ for the initial conditions 
$x \in (q,v) \in \tilde{U}_3(\azk) \subset \RR^\azk \times \RR^\azk$, then the 
first collision of this kind will occur between particle $\azk-1$ and $\azk$.

We denote the velocity difference of the particles $k+1$ and $k$ between the 
$\ell$--th and $(\ell+1)$--th collision after time $T$ by $\dvuo{k}{\ell}$. 
By definition of $\tilde{U}_3(\azk)$ we start with $\dvuo{k}{\ell} \geq 0$ for 
$k=\auf{1}{\azk-2}$ and we assume that $\dvuo{\azk-1}{0} \ll 0$. 
That assumption is justified, since we still can freely choose the initial 
velocity of the $n$--th particle.
Then we see
\BeqO
	\dvuo{\azk-1}{1}=-\dvuo{\azk-1}{0} \gg 0	
	\qmbox{,}	
	\dvuo{k}{1}= \dvuo{k}{0} \qquad(k=\auf{1}{\azk-3})
\EeqO
and 
\Beq
	\dvuo{\azk-2}{1} = 
	\dvuo{\azk-2}{0}+ \frac{2m_{\azk}}{m_{\azk-1}+m_{\azk}} 
	\dvuo{\azk-1}{0}\approx 
	\frac{2m_{\azk}}{m_{\azk-1}+m_{\azk}} 
	\dvuo{\azk-1}{0}\ll 0. \label{eq:approx}
\Eeq
If the next collision occurs between particle $\azk-2$ and $\azk-1$, then there is
\begin{align*}
	\dvuo{\azk-1}{2} &=  
	\dvuo{\azk-1}{1} + \frac{2m_{\azk-2}}{m_{\azk-2} +
	m_{\azk-1}}\dvuo{\azk-2}{1}\\
	& =  
	-\dvuo{\azk-1}{0} + \frac{2m_{\azk-2}}{m_{\azk-2}  +
	m_{\azk-1}} \le( \dvuo{\azk-2}{0} + 
	\frac{2m_{\azk}}{m_{\azk-1}+
	m_{\azk}}\dvuo{\azk-1}{0} \ri)\\
	& =  
	{\textstyle \frac{4 m_{\azk-2}m_{\azk}-(m_{\azk-2}+
	m_{n\azk-1})(m_{\azk-1}+m_{\azk})}{(m_{\azk-2}+
	m_{\azk-1})(m_{\azk-1}+m_{\azk})}}
	\dvuo{\azk-1}{0}+{\textstyle
	\frac{2m_{\azk-2}}{m_{\azk-2}+m_{\azk-1}}} \dvuo{\azk-2}{0}
	.
\end{align*}

Thus, the coefficient of $ \dvuo{n-1}{0} $ is positive, the numerator must be positive, i.e. it must apply
\begin{multline}
	0 < 
	4 m_{\azk-2}m_{\azk}-(m_{\azk-2}+m_{\azk-1})(m_{\azk-1}+m_{\azk}) \\
	= 
	(3m_{\azk-2} - m_{\azk-1})m_\azk - (m_{\azk-2}+m_{\azk-1})m_{\azk-1} .
	\label{eq:massen_rekursion_herleitung}
\end{multline}
Thus, we obtain the following condition for the $ n $--th mass
\Beq
	m_\azk > \frac{(m_{\azk-2}+m_{\azk-1})m_{\azk-1}}{3m_{\azk-2} - m_{\azk-1}}  
	\qmbox{for} 3m_{\azk-2} > m_{\azk-1}. 
	\label{eq:massen_rekursion_3}
\Eeq

Now we consider the quotient $ \frac{m_{\azk}}{m_{\azk-1}} $, then by \eqref{eq:schranken_fkten} and \eqref{eq:massen_rekursion_3} the result is
\BeqO
	\frac{m_{\azk}}{m_{\azk-1}} > 
	\frac{m_{\azk-2}+m_{\azk-1}}{3m_{\azk-2}-m_{\azk-1}} =
	\frac{ 1+\frac{m_{\azk-1}}{m_{\azk-2}} }{ 3-\frac{m_{\azk-1}}{m_{\azk-2}} }=
	f\le(\frac{m_{\azk-1}}{m_{\azk-2}}\ri) .
\EeqO

According to the preliminary note, we can find two initial conditions $m_1$ 
and $m_2$ in the interval $\big(1,(\azk+1)/\azk\big)$ for 
$\frac{m_{\azk}}{m_{\azk-1}}$  so that
\begin{multline*}
	1< \delta' = 
	{\textstyle \frac{(\azk-2)\le(\frac{m_2}{m_1}-1\ri)-
	2\frac{m_2}{m_1}}{(\azk-2)\le(\frac{m_2}{m_1}-1\ri)-2}} < 
	{\textstyle \frac{m_\azk}{m_{\azk-1}}} 
	< 2^{\azk-2}\le({\textstyle \frac{m_2}{m_1}}-1\ri)+1 = 
	\delta = (1+\ee)^{1/(\azk-1)} .
\end{multline*}
By iteration, using \eqref{eq:approx}, we have for $\azk-1 $ collisions negative values for $\dvuo{\ell}{\azk-1}$ ($\ell=2,\ldots,\azk-1$) and we have $\dvuo{\azk-1}{0}\ll 0$.
Moreover, $\dvuo{1}{\azk-1}\gg 0$.

This means, however, that for the particles $\auf{2}{\azk}$ by Case 2) a neighbourhood $U_2(\azk-1)$ can be found, so that we have ${\azk-1 \choose 2}$ collisions. Since $\azk-1+{\azk-1 \choose 2}={\azk \choose 2}$, we proved the assertion.
\end{proof}

\section{Numerical Example}

By the proof of Theorem \ref{thm:maintheorem}, Part 3) we can also produce 
${\azk+1 \choose 3}$ collisions with mass distribution 
$m_1 > \ldots > m_{k-1} >\; m_k \;< m_{k+1} < \ldots < m_\azk$.

If we consider \eqref{eq:massen_rekursion_herleitung}, then we can rewrite this equation to 
\BeqO
	4 m_{\azk-2}m_{\azk} > (m_{\azk-2}+m_{\azk-1})(m_{\azk-1}+m_{\azk})
\EeqO
and we see that this equation is symmetric in $m_{\azk-2}$ and $m_{\azk}$.

We can add masses, alternating between left and right, so we get ${\azk+1 \choose 3}$ collisions (see Figure \ref{fig:example}).

\begin{figure}[h]
\begin{picture}(0,0)%
\includegraphics{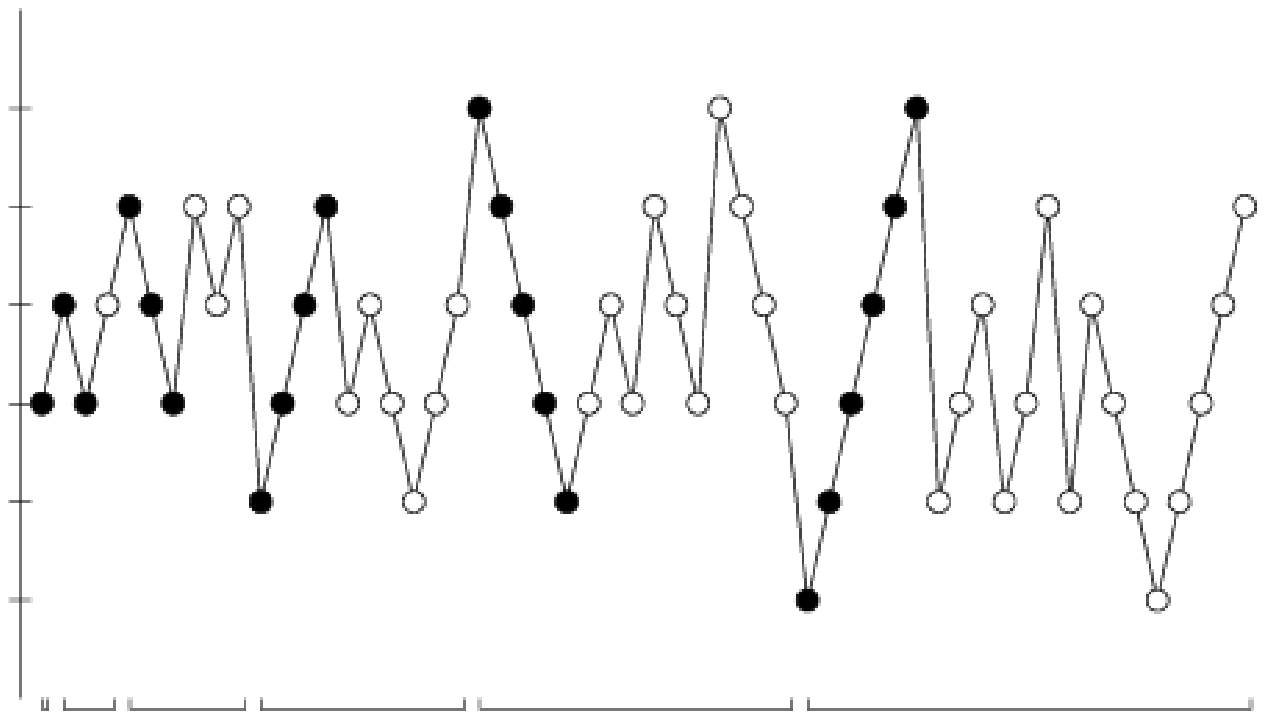}%
\end{picture}%
\setlength{\unitlength}{4144sp}%
\begingroup\makeatletter\ifx\SetFigFontNFSS\undefined%
\gdef\SetFigFontNFSS#1#2#3#4#5{%
  \reset@font\fontsize{#1}{#2pt}%
  \fontfamily{#3}\fontseries{#4}\fontshape{#5}%
  \selectfont}%
\fi\endgroup%
\begin{picture}(5753,3472)(-94,-2621)
\put(-79,344){\makebox(0,0)[rb]{\smash{{\SetFigFontNFSS{10}{12.0}{\familydefault}{\mddefault}{\updefault}{\color[rgb]{0,0,0}$q_6=q_7$}%
}}}}
\put(-79,-106){\makebox(0,0)[rb]{\smash{{\SetFigFontNFSS{10}{12.0}{\familydefault}{\mddefault}{\updefault}{\color[rgb]{0,0,0}$q_5=q_6$}%
}}}}
\put(-79,-556){\makebox(0,0)[rb]{\smash{{\SetFigFontNFSS{10}{12.0}{\familydefault}{\mddefault}{\updefault}{\color[rgb]{0,0,0}$q_4=q_5$}%
}}}}
\put(-79,-1006){\makebox(0,0)[rb]{\smash{{\SetFigFontNFSS{10}{12.0}{\familydefault}{\mddefault}{\updefault}{\color[rgb]{0,0,0}$q_3=q_4$}%
}}}}
\put(-79,-1456){\makebox(0,0)[rb]{\smash{{\SetFigFontNFSS{10}{12.0}{\familydefault}{\mddefault}{\updefault}{\color[rgb]{0,0,0}$q_2=q_3$}%
}}}}
\put(-79,-1906){\makebox(0,0)[rb]{\smash{{\SetFigFontNFSS{10}{12.0}{\familydefault}{\mddefault}{\updefault}{\color[rgb]{0,0,0}$q_1=q_2$}%
}}}}
\put(116,-2770){\makebox(0,0)[b]{\smash{{\SetFigFontNFSS{10}{12.0}{\familydefault}{\mddefault}{\updefault}{\color[rgb]{0,0,0}${2 \choose 2}$}%
}}}}
\put(316,-2770){\makebox(0,0)[b]{\smash{{\SetFigFontNFSS{10}{12.0}{\familydefault}{\mddefault}{\updefault}{\color[rgb]{0,0,0}${3 \choose 2}$}%
}}}}
\put(766,-2770){\makebox(0,0)[b]{\smash{{\SetFigFontNFSS{10}{12.0}{\familydefault}{\mddefault}{\updefault}{\color[rgb]{0,0,0}${4 \choose 2}$}%
}}}}
\put(1566,-2770){\makebox(0,0)[b]{\smash{{\SetFigFontNFSS{10}{12.0}{\familydefault}{\mddefault}{\updefault}{\color[rgb]{0,0,0}${5 \choose 2}$}%
}}}}
\put(2816,-2770){\makebox(0,0)[b]{\smash{{\SetFigFontNFSS{10}{12.0}{\familydefault}{\mddefault}{\updefault}{\color[rgb]{0,0,0}${6 \choose 2}$}%
}}}}
\put(4616,-2770){\makebox(0,0)[b]{\smash{{\SetFigFontNFSS{10}{12.0}{\familydefault}{\mddefault}{\updefault}{\color[rgb]{0,0,0}${7 \choose 2}$}%
}}}}
\end{picture}%

\caption{Numerical example of the theorem for $\azk=7$ particles and 
mass distribution $m_1 > \ldots > m_3$, $m_3 < \ldots < m_7$. 
The $(\azk-1)$ collisions from the induction step are black-colored.} \label{fig:example}
\end{figure} 

%
%
\end{document}